\documentclass[a4paper]{article}
\pdfoutput=1
\usepackage[english]{babel}
\usepackage{ucs}
\usepackage{url}
\usepackage[utf8x]{inputenc}
\usepackage{amsmath, amssymb, amsthm, amsfonts}
\usepackage{graphicx}
\usepackage{float}

\newtheorem{lemma}{Lemma}[section]

\newtheorem{proposition}{Proposition}[section]

\DeclareMathOperator{\acos}{acos}

\DeclareMathOperator{\dist}{dist}
\DeclareMathOperator{\normalize}{normalize}

\renewcommand{\vec}[1]{\mathbf{#1}}

\hypersetup{colorlinks,urlcolor=blue,citecolor=black,linkcolor=black,bookmarks=false, pdfauthor={Jorn H. Baayen},
            pdftitle={Trajectory tracking control of kites with system delay},
            pdfkeywords={kite, control, delay, dde, delay differential equation, geodesic control, great circle control, control on the sphere}}

\title{Trajectory tracking control of kites \\ with system delay}

\date{December 27, 2012}

\author{
  J. H. Baayen\thanks{Baayen \& Heinz GmbH, Sekr. ER2-1, Hardenbergstra\ss e 36a, 10623 Berlin, Germany, jorn.baayen@baayen-heinz.com.}
}

\begin{document}

\maketitle

\begin{abstract}
A previously published algorithm for trajectory tracking control of tethered wings, i.e. kites, is updated in light of recent experimental evidence.  The algorithm is, furthermore, analyzed in the framework of delay differential equations.  It is shown how the presence of system delay influences the stability of the control system, and a methodology is derived
for gain selection using the Lambert $W$ function.  The validity of the methodology is demonstrated with simulation results.  The analysis sheds light on previously poorly understood stability problems. 
\end{abstract}

\section{Introduction}

During the past two decades researchers have taken an increasing interest in the industrial applications of tethered wings, i.e., kites.
Two well-known applications include the towing of ships \cite{Naaijen2006} and wind- to electrical power conversion \cite{Loyd1980,Canale2007}.  An integral part of both technologies is autonomous kite flight 
along a prescribed trajectory, relative to the earth surface attachment point.  Various control principles have been proposed, see, e.g, \cite{Fechner2012,Erhard2012,Baayen2012,Fagiano2009,Williams2007,Ilzhofer2006}.  This paper discusses the influence of system delays on trajectory tracking performance.  The issue of system delays has been raised before in \cite{Fechner2012}, where a wireless link is part of the control loop, and in \cite{Erhard2012}, where actuator delay is identified.  Up to now, however, no rigorous analysis of the effect of delay on system performance has been carried out.  

\begin{figure}[h]

\centering

\includegraphics[width=200pt]{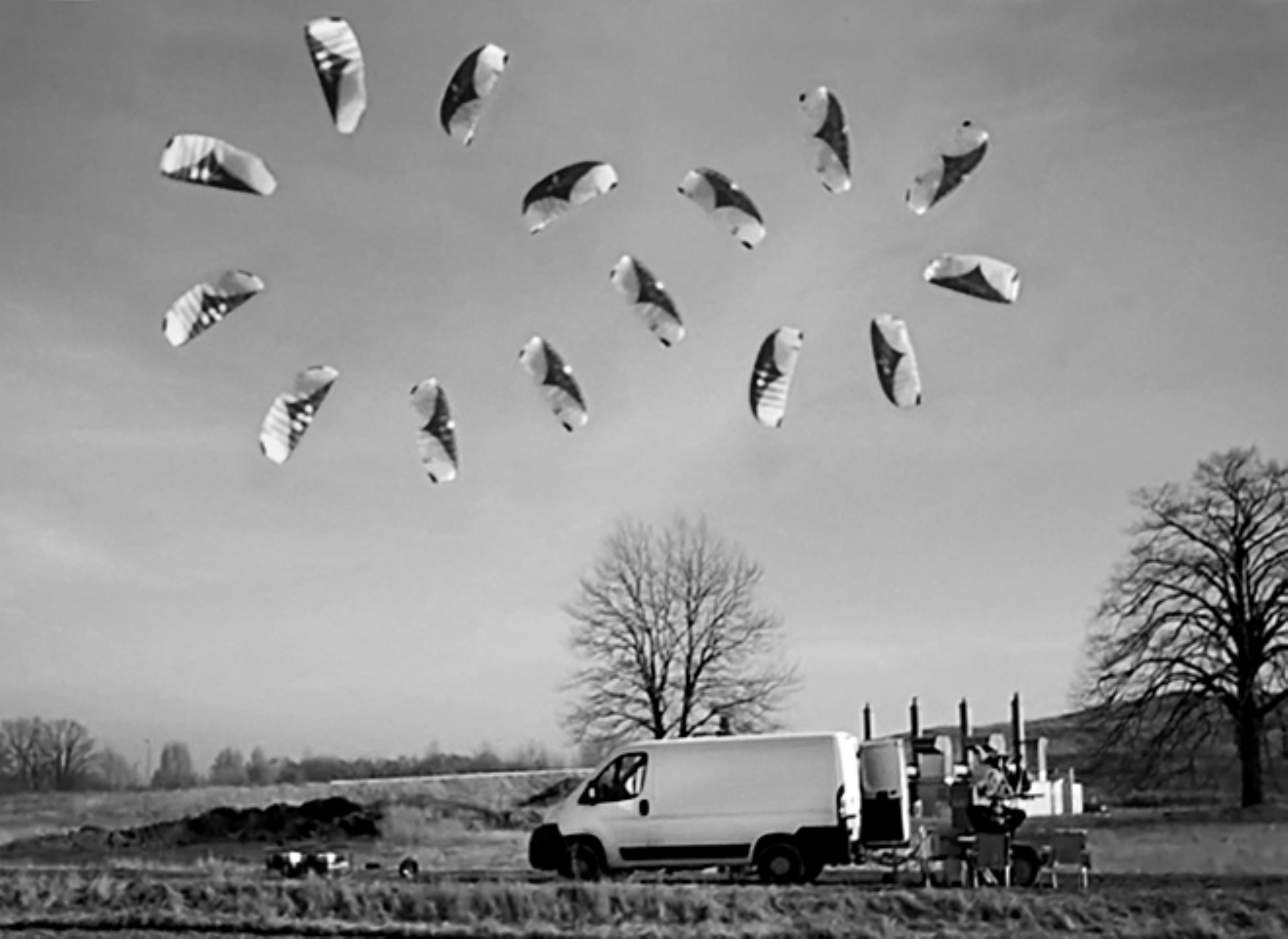}
\caption{A $5$ m$^2$ de-power kite on a $40$ m line in automated flight. Composite photo taken near Schöneiche, Germany, on February 22, 2012.}
\label{fig:kite-montage}

\end{figure}

A small kite such as the one shown in Figure \ref{fig:kite-montage} is highly sensitive to control inputs, and hence to system delay.  The ability to fly such small wings smoothly with minimal control expenditure therefore hinges on an accurate treatment of their fast dynamics.  We use the theory of delay differential equations to incorporate system delay into the stability analysis of a kite trajectory tracking control algorithm.

We base our analysis on the controller previously laid out in \cite{Baayen2012}.  Since its initial publication, the algorithm has been applied in practice (Figure \ref{fig:kite-montage}), and in the process, has undergone several revisions. In the next section, we commence by summarizing our algorithm in its most recent form.  We then discuss the stability of an elementary linear differential equation with delay.  Noting that our inner loop controller shapes the system into a delayed linear system of the same type, we proceed to identify its region of stability and, in a certain sense, optimal control gain.  In the final section we demonstrate the applicability of the result in simulation.

\section{Trajectory tracking control}\label{section:control}

In \cite{Baayen2012} we showed how the turning, or path, angle of the kite trajectory can be used to formulate a cascaded single-input single-output control problem.  In this section we briefy recapitulate this result, and present several modifications which have turned out to improve performance in practice.

First of all, the presence of the kite's tether constrains its movement to a sphere centered on the tether's earth surface attachment point.  In our analysis, the radius of this sphere turns out to be irrelevant, and therefore we restrict our attention to the unit sphere centered on the origin.

Restricted to the unit sphere, we can speak of geodesic, or great-circle, distances
\[
\dist(\vec{p},\vec{q}) := \acos (\vec{p} \cdot \vec{q}),
\]
as well as of tangent planes.  After defining a basis in the tangent bundle, we can measure the angle between any tangent vector and the primary basis vector.  For vectors tangent to a curve, this angle is known as the turning angle \cite{Gray1997} -- in this paper denoted as $\theta$ (Figure \ref{fig:turning-angle}).

\begin{figure}[h]

\centering

\includegraphics[width=250pt]{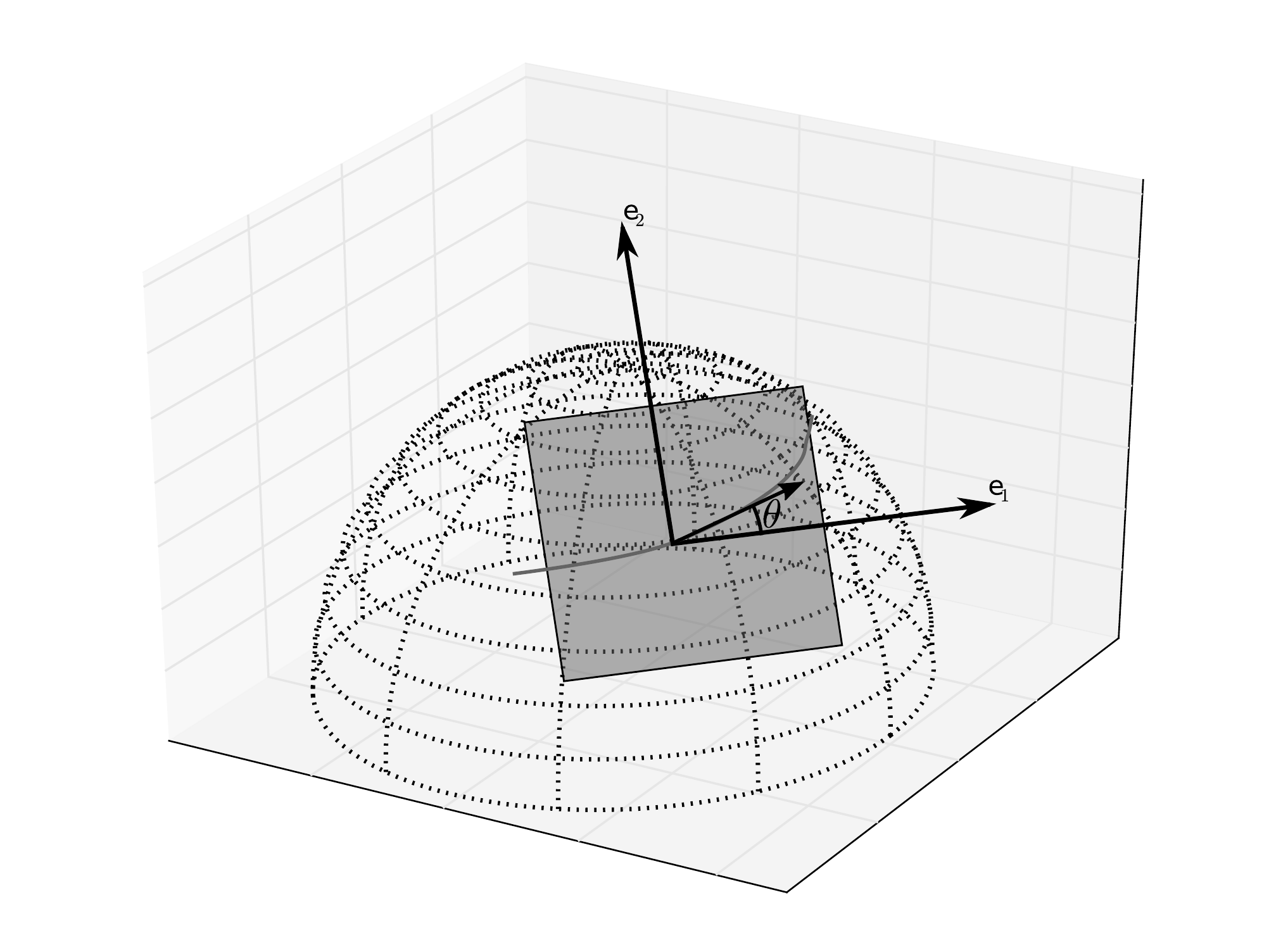}
\caption{The notion of turning, or path, angle $\theta$.}
\label{fig:turning-angle}

\end{figure}

In \cite{Baayen2012} we defined the steering control input, $u$, as the length difference of the kite steering lines.  We showed how this control input $u$ can be used to control the direction of flight, represented by the turning angle $\theta$.  Secondly, we used spherical geometry to derive an outer loop control law that selects a target turning angle, $\theta_t$, such as to guarantee convergence to the target trajectory.

A commonly prescribed target trajectory is the crosswind, lying figure of eight (Figure \ref{fig:kite-montage}).  This trajectory has a bounded winding number, thereby avoiding problems with lines winding around themselves and each other.

We now proceed to summarize the inner and outer loop control laws.

\subsection{Inner loop}

In \cite{Erhard2012} it is, based on experimental evidence, claimed that
\begin{equation}\label{eq:turning-rate-law}
\dot{\theta} = K V_a u,
\end{equation}
with proportionality factor $K$, airspeed $V_a$, and steering control input $u$.\footnote{In \cite{Erhard2012}, the turning angle $\theta$ is called the ``orientation angle'', denoted $\psi$.}  The
linear structure of Equation \ref{eq:turning-rate-law} obsoletes the incremental feedback-linearizing approach previously taken by us in \cite{Baayen2012}, and thereby simplifies matters significantly.

However, the airspeed $V_a$ is also a function of the turning angle, and hence also a function of the steering control input $u$.  Implicit in Equation \ref{eq:turning-rate-law} is the assumption that the influence of the control input $u$ acts on the turning rate $\dot{\theta}$ more quickly than it acts on the airspeed $V_a$.  We return to this time-scale separation in Section \ref{section:dde-control}.

\begin{proposition}
Consider the system given by Equation \ref{eq:turning-rate-law}. The control law
\[
u := \frac{1}{K V_a} (\dot{\theta_t} + P(\theta_t - \theta)),
\]
with gain $P \in \mathbb{R}^+$, renders the Lyapunov function
\[
V := \frac{1}{2}(\theta_t - \theta)^2
\]
decreasing along the flow.
\end{proposition}
\begin{proof}
Expand $dV/dt$.
\end{proof}

Future work on the inner loop may focus on the incorporation of actuator dynamics using, e.g., backstepping \cite{Khalil1996}.

\subsection{Outer loop}

As the kite moves along its trajectory, $\boldsymbol\gamma$, the desired target point $\boldsymbol\gamma_t(s)$ must slide along on the target trajectory, $\boldsymbol\gamma_t$.  To this end, we continuously locate the target point closest to the present kite location, i.e.,
\begin{equation}\label{eq:closest-target-point}
s(t):=\mbox{argmin}_{s} \dist(\boldsymbol\gamma(t),\boldsymbol\gamma_t(s)).
\end{equation}
This optimization problem is not convex.  But, since we desire to track a smoothly changing target location, we welcome local minima.  A simple gradient descent \cite{Fletcher1987} performs well for this purpose.  When initialized with the solution from the previous timestep, the method converges within a handful of iterations. Unlike the method we originally proposed in \cite{Baayen2012}, this method does not suffer from drift of the target location.

In the following we will, for brevity, identify $\boldsymbol\gamma_t$ with the composition $\boldsymbol\gamma_t \circ s$. Once a target point has been determined using Equation \ref{eq:closest-target-point}, a target turning angle needs to be selected.  To this end, let $\vec{T}_t$ denote the unit vector tangent to the target trajectory. Let $\vec{T}:=\operatorname{R}^T \vec{T}_t$, where the operator $\operatorname{R}$ rotates $\gamma$ to $\gamma_t$ along their connecting geodesic, so that $\vec{T}$ is an element of the tangent plane at $\boldsymbol\gamma$.
Let $\vec{t}:=\boldsymbol\gamma \times \vec{T}$, so that the vectors $\vec{T}$
and $\vec{t}$ form a basis of the tangent plane at $\boldsymbol\gamma$. Finally, denote
\[
\boldsymbol\gamma_t^{\perp}:=\boldsymbol\gamma_t-(\boldsymbol\gamma_t\cdot\boldsymbol\gamma)\boldsymbol\gamma.
\]

By selecting the target flight direction as an appropriately weighted sum of the ``direction along the target'' $\vec{T}$, and the ``direction towards the target'' $\vec{t}$, we obtain tracking convergence:

\begin{proposition}[Baayen, 2012 \cite{Baayen2012}]
The turning angle $\theta_t$ determined by
\begin{equation}\label{eq:outer-loop-control-law}
\cos \theta_t \vec{e}_1 + \sin \theta_t \vec{e}_2 = \normalize(\vec{T} + L (\vec{t} \cdot \boldsymbol\gamma_t^{\perp}) \vec{t}),
\end{equation}
with gain $L \in \mathbb{R}^+$, renders the geodesic distance
\[
W:=\dist(\boldsymbol\gamma,\boldsymbol\gamma_t)
\]
decreasing along the flow.\footnote{In the original paper, the required normalization is missing.  For practical applications this does not matter, as the trigonometric computation of $\theta_t$ is only a function of the direction of the vector, not of its norm.  Thanks to Dipl.-Ing. Claudius Jehle for pointing this out.}
\end{proposition}

\section{A delay differential equation}\label{section:dde}

We now introduce an elementary linear delay differential equation, which in Section \ref{section:dde-control} will prove to be a valuable tool for analyzing the stability of our kite control system in face of delay.

Let
\begin{equation}\label{eq:dde}
\frac{dy}{dt}(t)=ky(t-\tau),
\end{equation}
with $k \in \mathbb{R}$, a linear differential equation with delay $\tau$.  With the change of variables $s=t/\tau$ and $x(s)=y(\tau s)$, we obtain the equation
in its canonical form
\begin{equation}\label{eq:canonical-dde}
\frac{dx}{ds}(s)=ax(s-1),
\end{equation}
with $a=\tau k$. 

Unlike its relative without delay, Equation \ref{eq:canonical-dde} admits oscillatory solutions.  For example, when $a=-\pi/2$,
it is readily verified that
\[
x=\sin\left(\frac{\pi}{2} s\right)
\]
provides such a solution.  Hence it is not the case that a negative coefficient in Equations \ref{eq:dde} or \ref{eq:canonical-dde} guarantees convergence to the zero-solution.  We must therefore consider these equations in more detail.

It turns out that the solution space of Equation \ref{eq:canonical-dde} is infinite-dimensional \cite{Erneux2009}.  This circumstance is hardly surprising, considering that we need to specify initial conditions over an interval. Any solution can be
written as
\[
x = \sum_n \lambda_n e^{\sigma_n s},
\]
with the coefficients $\lambda_n$ determined by the initial conditions, and $\sigma_n$ the roots of the characteristic equation
\[
\sigma - a e^{-\sigma} = 0.
\]
The roots of the characteristic equation, $\sigma_n$, are given by the set-valued Lambert $W$ function:
\[
\{\sigma_n\}=W(a).
\]
The real part of the first several branches of the Lambert $W$ function is shown in Figure \ref{fig:lambert-w}.

\begin{figure}[h]

\centering

\includegraphics[width=250pt]{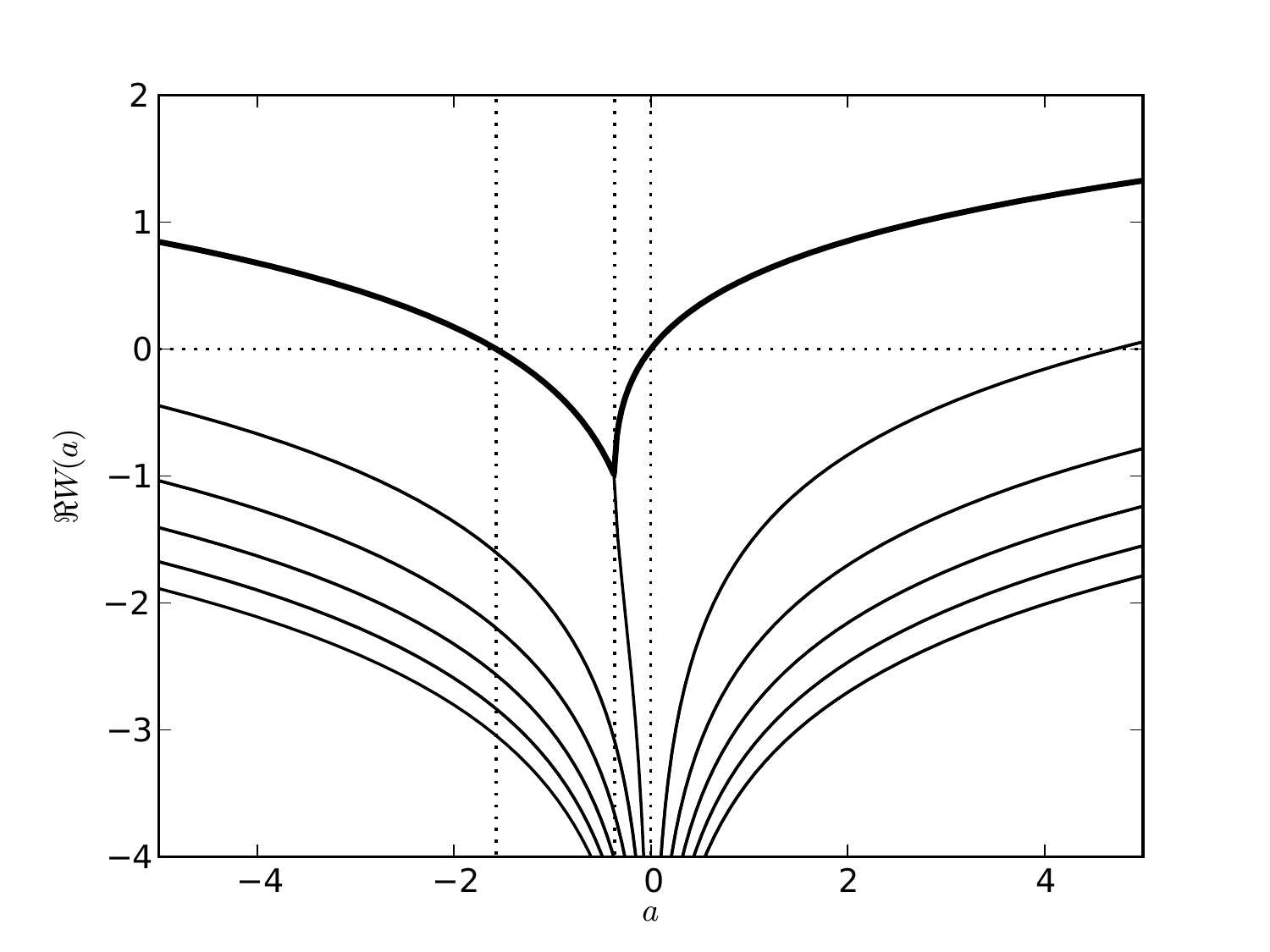}
\caption{Real part of the Lambert $W$ function, branches $-5$ to $5$, with the principal branch highlighted in bold.}
\label{fig:lambert-w}

\end{figure}

We proceed to analyze the stability of Equation \ref{eq:canonical-dde}.  Let $W_k$ denote the $k$th branch of the Lambert $W$ function.   Figure \ref{fig:lambert-w} hints at the following Lemma:

\begin{lemma}[Shinozaki and Mori, 2006 \cite{Shinozaki2006}]\label{lemma:shinozaki2006}
The real part of the principal branch of the Lambert $W$ function is no less than that of any other branch, i.e.,
\[
\max\{\Re(W_k(z)) | k = 0, \pm 1, . . . , \pm \infty\} = \Re(W_0(z)).
\]
\end{lemma}

\begin{lemma}\label{lemma:dde-stability}
The zero-solution of Equation \ref{eq:canonical-dde} is globally stable if
\[
-\pi/2 \leq a \leq 0.
\]
\end{lemma}
\begin{proof}
The proof follows directly from Lemma \ref{lemma:shinozaki2006}, Figure \ref{fig:lambert-w}, and the observations that $W(-\pi/2) \ni \imath \pi/2$ and $W(0) \ni 0$.
\end{proof}

\begin{lemma}\label{lemma:dde-optimality}
The quickest worst-case convergence to the zero-solution of Equation \ref{eq:canonical-dde} is guaranteed if
\[
a = -\frac{1}{e},
\]
irrespective of the initial conditions.
\end{lemma}
\begin{proof}
The proof follows directly from Lemma \ref{lemma:shinozaki2006}, Figure \ref{fig:lambert-w}, and from the fact that $W$ is not differentiable at $-1/e$ \cite{Corless1996}.
\end{proof}

\section{Trajectory tracking control with delay}\label{section:dde-control}

In the present section we modify the results from Section \ref{section:control} to incorporate a system delay $\tau$. 

First of all, we recall the time-scale separation of the turning rate $\dot{\theta}$ and the airspeed $V_a$ implicit in Equation \ref{eq:turning-rate-law}.  The turning angle $\theta$ specifies the direction of flight.  Considering that the direction of flight acts on the position indirectly through integration, we can expect another time-scale separation between the flight direction and the position dynamics.

In the following we show how this structure can be exploited to analyze stability in face of system delays. 

\subsection{Inner loop}

The outer loop control law given by Equation \ref{eq:outer-loop-control-law} is a function of position only.  By time-scale separation, this implies that the resulting target turning angle $\theta_t$ varies slowly compared to the time constant of the inner loop.  In the present discussion of the inner loop we therefore assume that $d\theta_t/dt=0$.  

Exploiting this assumption and following Section \ref{section:dde}, we obtain the following stability criterion:

\begin{proposition}\label{prop:stability-criterion}
Consider the turning rate law with delay $\tau$:
\begin{equation}\label{eq:delayed-turning-rate-law}
\dot{\theta}(t) = K V_a(t) u(t - \tau),
\end{equation}
cf. Equation \ref{eq:turning-rate-law}. Assume that $d\theta_t/dt=0$, and let $P \in (0, \frac{\pi}{2\tau})$.  Then the control law
\begin{equation}\label{eq:delayed-control-command}
u := \frac{1}{K V_a} P(\theta_t - \theta)
\end{equation}
renders $\theta \to \theta_t$.  Furthermore, this convergence is optimal in the sense of Lemma \ref{lemma:dde-optimality} when
\[
P=\frac{1}{e\tau}.
\]
\end{proposition}
\begin{proof}
Substituting Equation \ref{eq:delayed-control-command} into the system given by Equation \ref{eq:delayed-turning-rate-law}, and using our assumption that $d\theta_t/dt=0$, we obtain the controlled system
\begin{equation}\label{eq:delayed-controlled-inner-system}
\frac{d}{dt}(\theta_t - \theta)(t) = -P(\theta_t - \theta)(t - \tau).
\end{equation}
From Lemma \ref{lemma:dde-stability}, we obtain that the controlled system given by Equation \ref{eq:delayed-controlled-inner-system} is stable if $-\pi/2~\leq~-P\tau~\leq~0$.
\end{proof}

This quantifies what we already know from practice, namely, that the inner loop control gain $P$ cannot be arbitrarily large.  Proposition \ref{prop:stability-criterion} aids us in selecting a gain appropriate for the specific system implementation at hand.

\subsection{Outer loop}

The outer loop treats the target direction $\theta_t$ as a virtual control input, capable of being commanded pseudo-instantaneously.  This assumption is motivated by the time-scale separation.  The target turning angle $\theta_t$ will be commanded to the inner loop without delay within the same control software cycle.

Furthermore, we assume the system delay $\tau$ to be insignificant compared to the time scale of the position dynamics.  As we will be able to verify with the simulation results from the next section, the system delay is typically a small fraction of the flight time of a figure of eight.\footnote{In contrast, a kite can be made to spin rapidly about its tether.}  Hence, we do not include delay in our analysis of the outer loop.

\section{Simulation results}

In \cite{Baayen2012-2} we show how one may simulate a kite system using potential flow computations and a discretized tether.  Using an identical system setup, we now present several simulation results. 

We start out by presenting a reference result without delay.  Selecting an outer loop gain of $L=5$ and an inner loop gain of $P=10$, the target trajectory is tracked, for all practical purposes, perfectly. See Figures \ref{fig:simulation-trajectory-D0L5P10} and \ref{fig:simulation-angle-D0L5P10}. 

By introducing a delay of $\tau=0.2$ s, the gain $P=10$ lies outside of the range given by Proposition \ref{prop:stability-criterion}.  As expected, this is reflected in the instability of the turning angle. See Figures \ref{fig:simulation-trajectory-D200L5P10} and \ref{fig:simulation-angle-D200L5P10}. 

Finally, by choosing $L=5$ and the optimal inner loop gain $P=1/(e\tau)$, we obtain the convergent trajectory and stable turning angle tracking of Figures \ref{fig:simulation-trajectory-D200L5Popt} and \ref{fig:simulation-angle-D200L5Popt}.  We note, however, that trajectory tracking performance suffers somewhat at the expense of stable turning angle tracking. 

We conclude that Proposition \ref{prop:stability-criterion} provides a useful tool for understanding, and eliminating, delay-induced instabilities observed in kite control systems. 

\begin{figure}[p]
\begin{minipage}[t]{0.47\linewidth}

\centering

\includegraphics[trim=0mm 0mm 0mm 0mm,width=160pt]{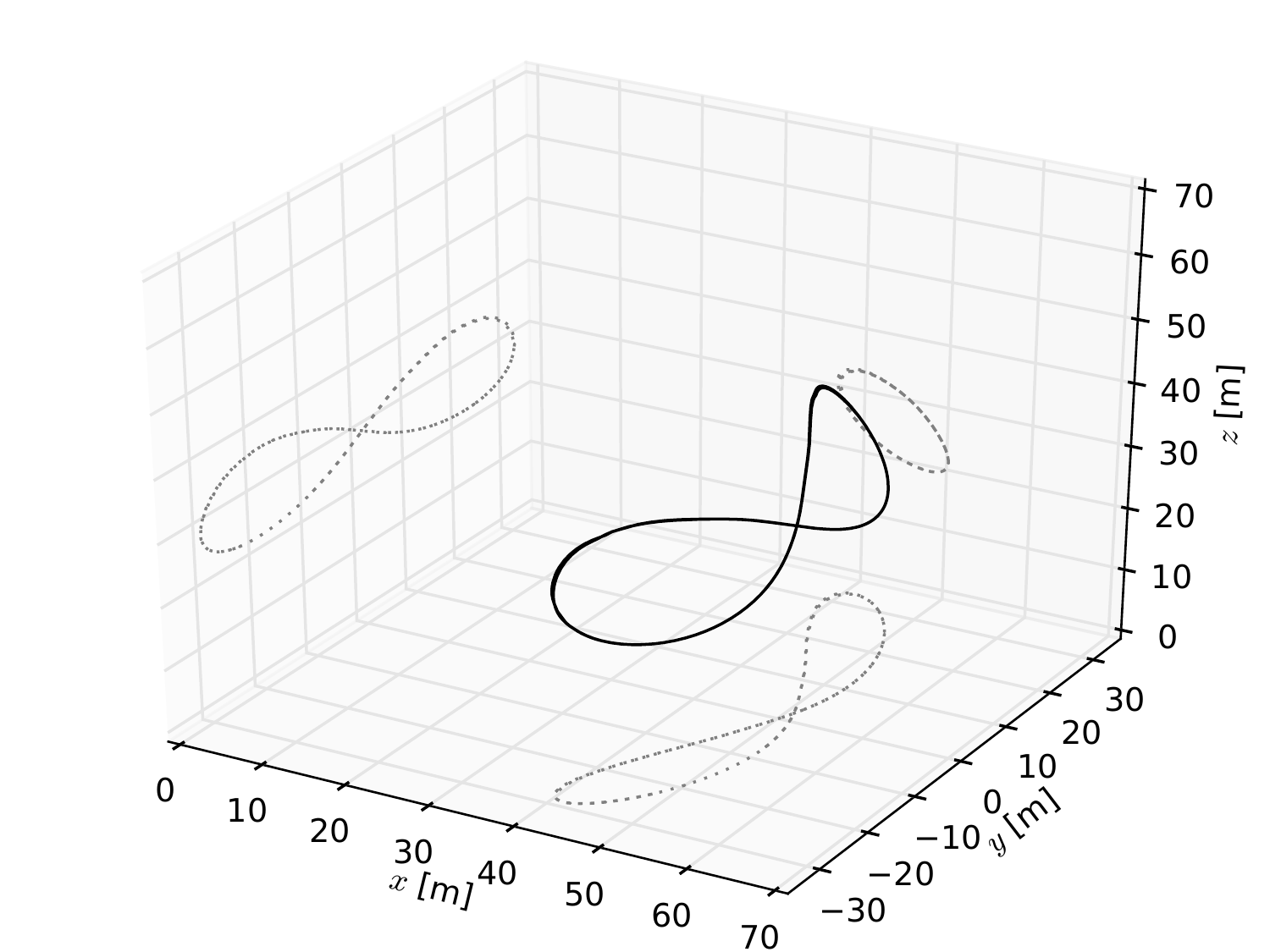}
\caption{Trajectory with $\tau=0$ s, $L=5$, and $P=10$.}
\label{fig:simulation-trajectory-D0L5P10}

\end{minipage}
\hspace{0.3cm}
\begin{minipage}[t]{0.47\linewidth}

\centering

\includegraphics[trim=0mm 0mm 0mm 0mm,width=160pt]{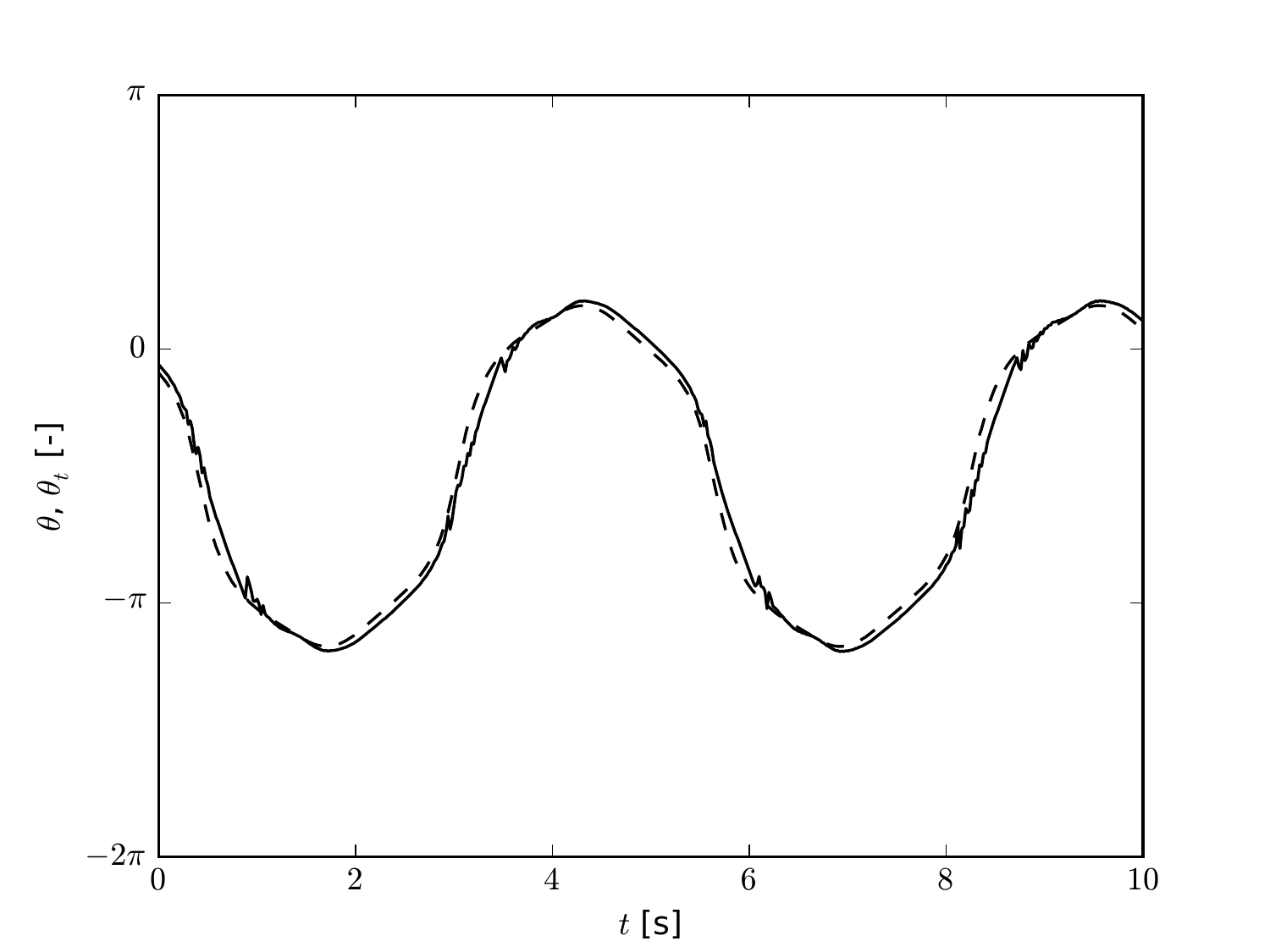}
\caption{Turning angle (solid) with target (dashed), $\tau=0$ s, $L=5$, and $P=10$.}
\label{fig:simulation-angle-D0L5P10}

\end{minipage}
\end{figure}

\begin{figure}[p]
\begin{minipage}[t]{0.47\linewidth}

\centering

\includegraphics[trim=0mm 0mm 0mm 0mm,width=160pt]{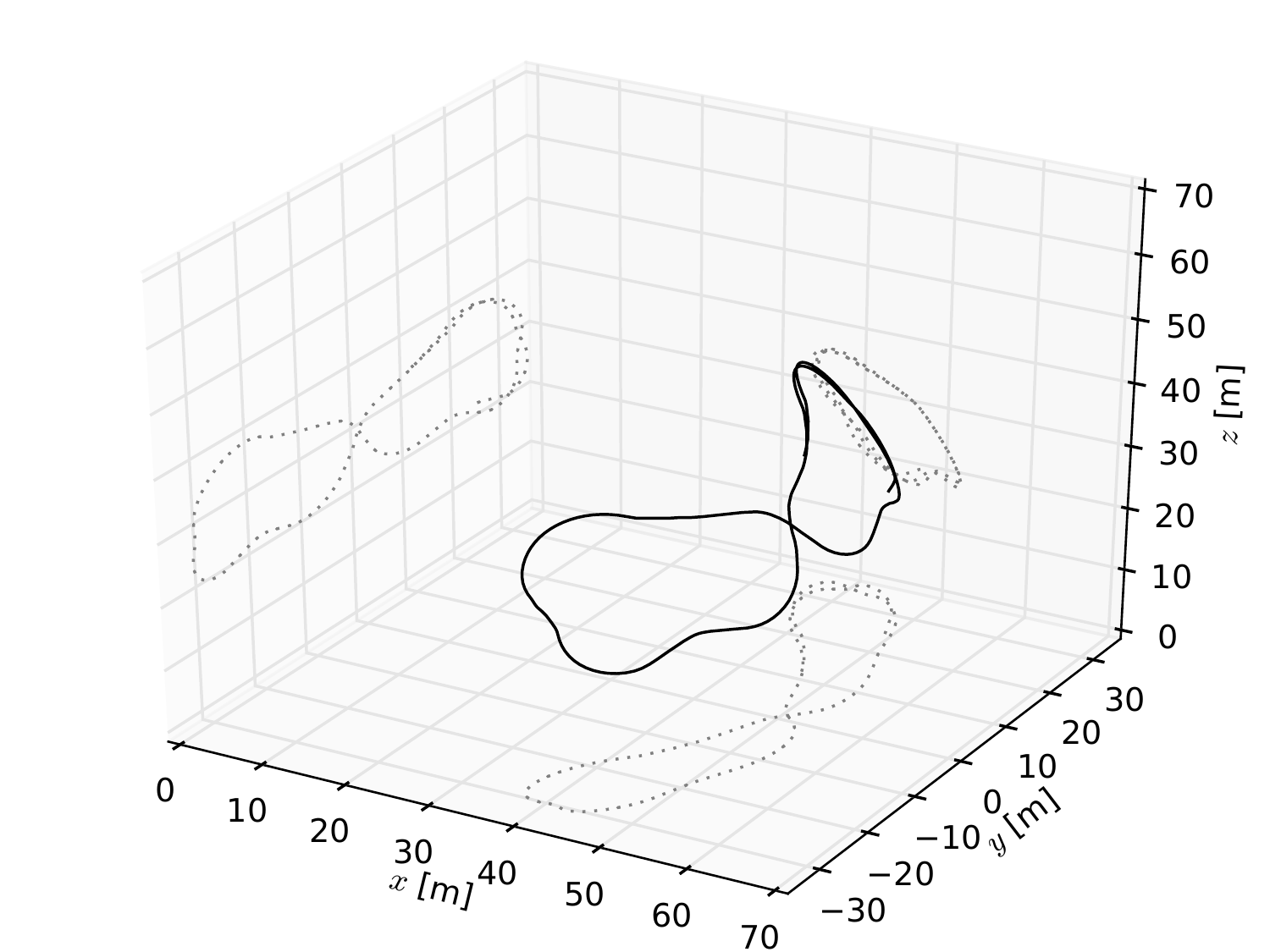}
\caption{Trajectory with $\tau=0.2$ s, $L=5$, and $P=10$.}
\label{fig:simulation-trajectory-D200L5P10}

\end{minipage}
\hspace{0.3cm}
\begin{minipage}[t]{0.47\linewidth}

\centering

\includegraphics[trim=0mm 0mm 0mm 0mm,width=160pt]{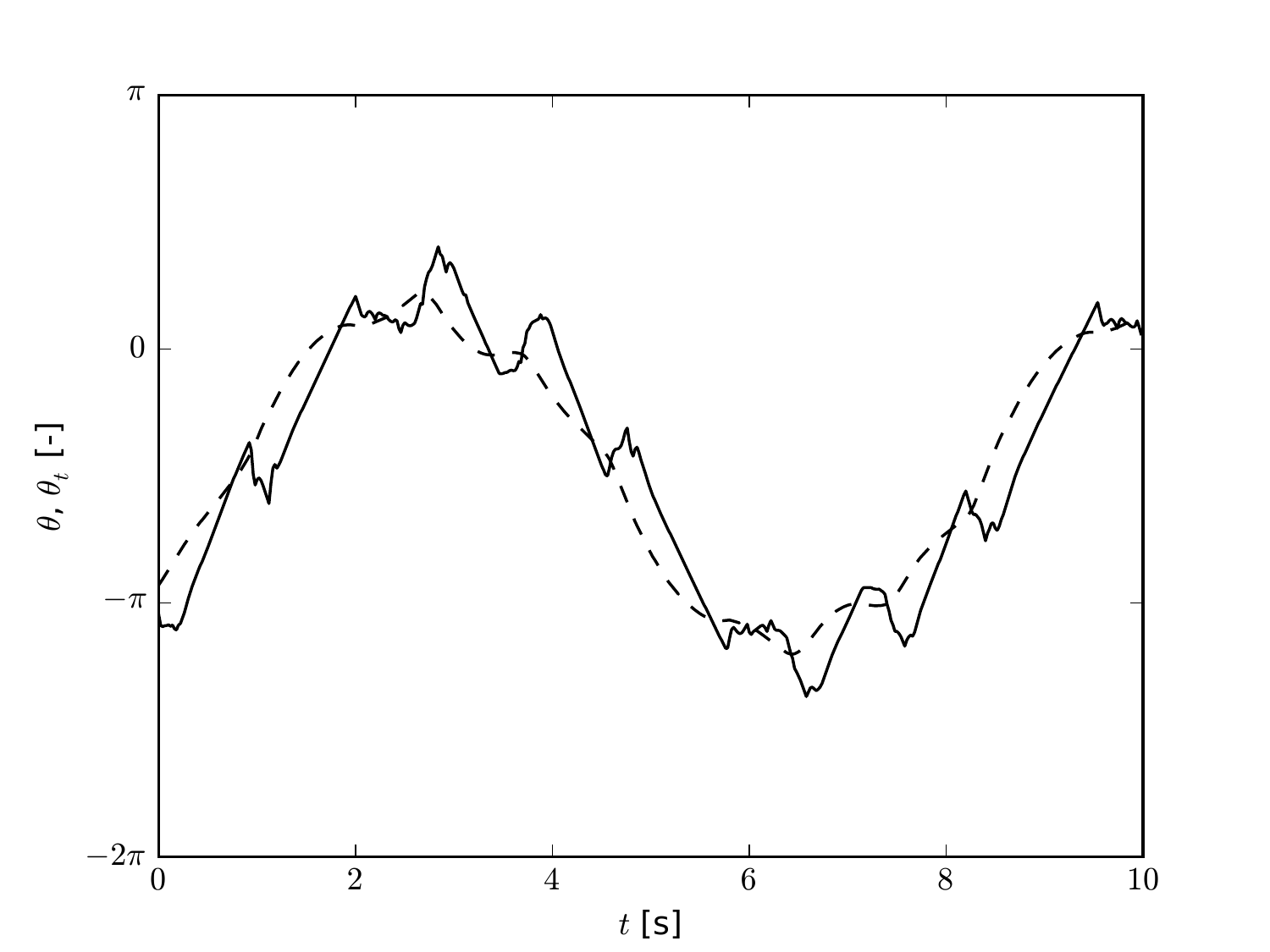}
\caption{Turning angle (solid) with target (dashed), $\tau=0.2$ s, $L=5$, and $P=10$.}
\label{fig:simulation-angle-D200L5P10}

\end{minipage}
\end{figure}

\begin{figure}[p]
\begin{minipage}[t]{0.47\linewidth}

\centering

\includegraphics[trim=0mm 0mm 0mm 0mm,width=160pt]{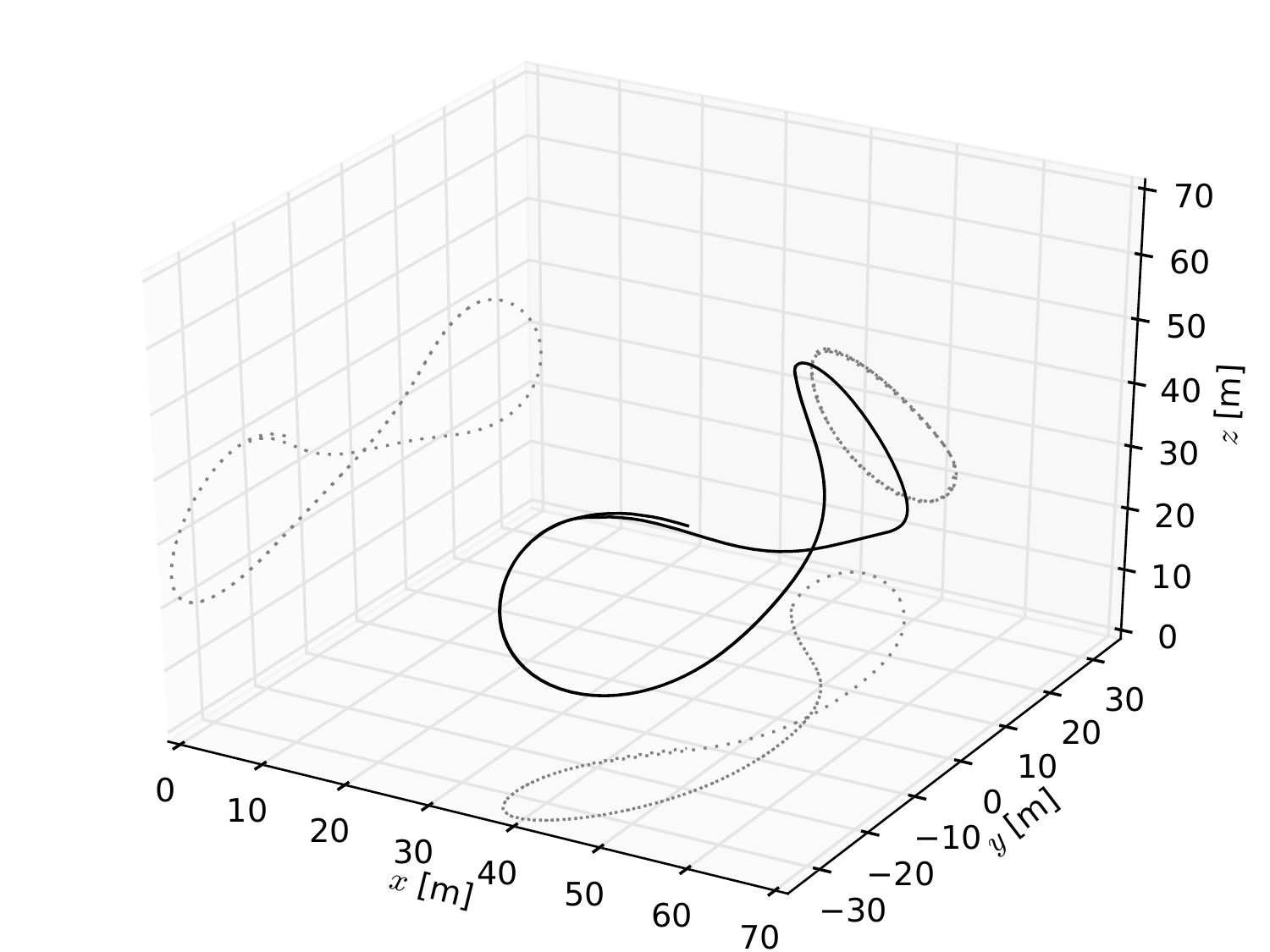}
\caption{Trajectory with $\tau=0.2$ s, $L=5$, and $P=1/(e\tau)$.}
\label{fig:simulation-trajectory-D200L5Popt}

\end{minipage}
\hspace{0.3cm}
\begin{minipage}[t]{0.47\linewidth}

\centering

\includegraphics[trim=0mm 0mm 0mm 0mm,width=160pt]{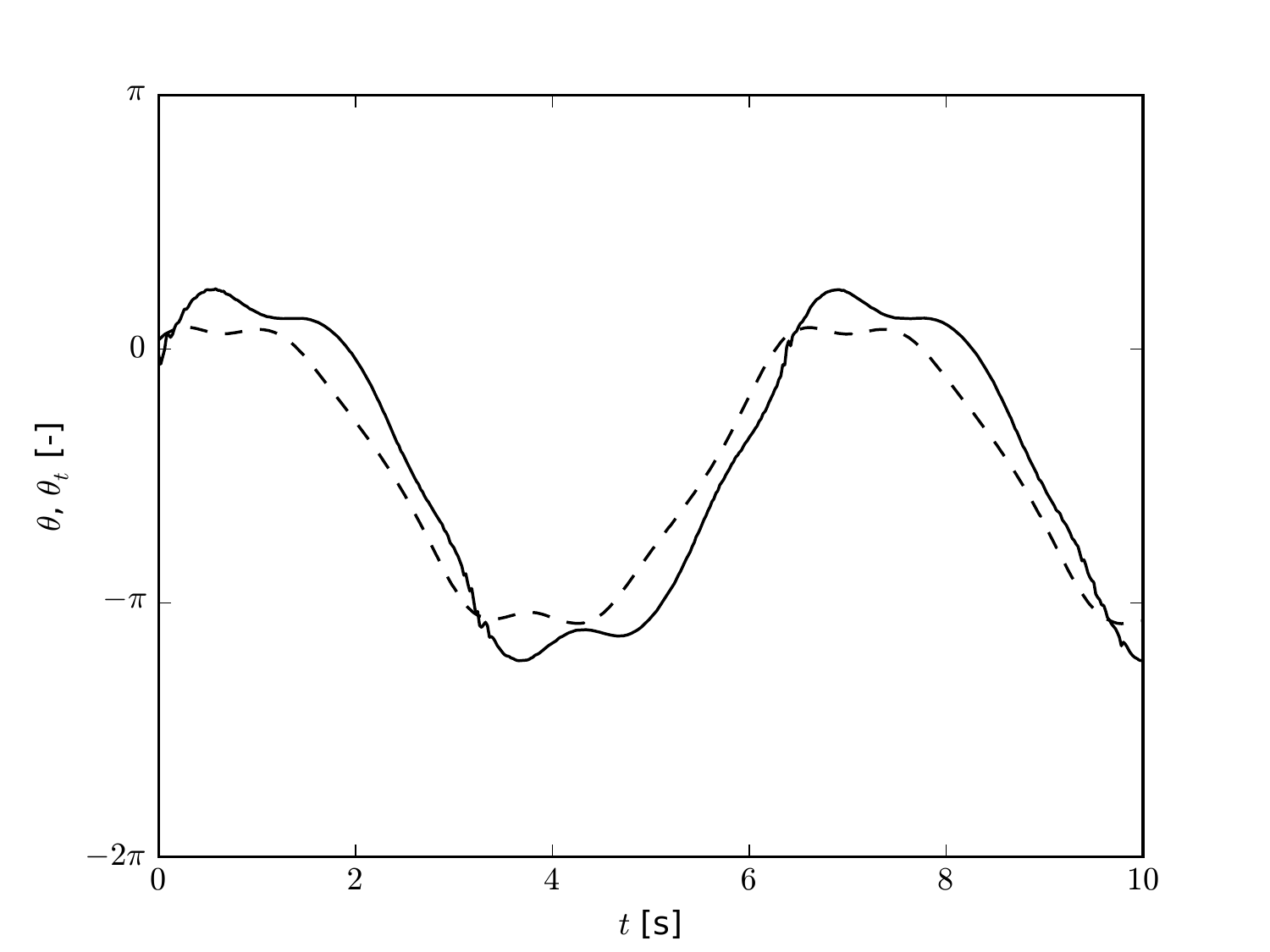}
\caption{Turning angle (solid) with target (dashed), $\tau=0.2$ s, $L=5$, and $P=1/(e\tau)$.}
\label{fig:simulation-angle-D200L5Popt}

\end{minipage}
\end{figure}

\section{Concluding remarks}

In this paper we have presented, based on recent experimental evidence, a simplification of our previously published cascaded kite trajectory tracking control algorithm.  Furthermore, we have shown how the Lambert $W$ function can be used to analyze the stability
of kite trajectory tracking control systems with delay.  We demonstrated how such an analysis leads to a methodology
for gain selection in face of delay.  Finally, we demonstrated the validity of our approach with simulation results.

Future research may focus on the influence of actuator dynamics on inner control loop performance, and similarly, on the influence of inner loop dynamics on the performance of the outer loop.

\bibliography{kite-control}
\bibliographystyle{hplain}

\end{document}